%% file: root.tex
\documentclass[letterpaper, 10 pt, conference]{ieeeconf}  % Comment this line out
\IEEEoverridecommandlockouts                              % This command is only
\input{IEEE/Inputs/packages}
\input{IEEE/Inputs/commands}

\input{IEEE/Inputs/theories}
\input{IEEE/Inputs/acronyms}

\title{\LARGE \bf 
Piecewise Control Barrier Functions for
Stochastic Systems
}
\author{
Rayan Mazouz$^{1}$, Luca Laurenti$^{2}$, Morteza Lahijanian$^{1}$
\thanks{$^1$Authors are with the University of Colorado Boulder, USA
\texttt{\{firstname.lastname\}@colorado.edu}}
\thanks{$^2$Authors are with Delft University of Technology, Netherlands
\texttt{\{l.laurenti\}\@tudelft.nl}}
}

\begin{document}
\AddToShipoutPictureBG*{%
  \AtPageUpperLeft{%
    \hspace{12.5cm}%
    \raisebox{-1.5cm}{%
      \makebox[0pt][r]{To appear in the 64th IEEE Conference on Decision and Control (CDC 2025)}}}}

\maketitle
\thispagestyle{empty}
\pagestyle{plain}

\maketitle

\begin{abstract}
This paper presents a method for the simultaneous synthesis of a barrier certificate and a safe controller for discrete-time nonlinear stochastic systems. Our approach, based on \emph{piecewise stochastic control barrier functions}, reduces the synthesis problem to a minimax optimization, which we solve exactly using a dual linear program with zero gap. This enables the joint optimization of the barrier certificate and safe controller within a single formulation. The method accommodates stochastic dynamics with additive noise and a bounded continuous control set. The synthesized controllers and barrier certificates provide a formally guaranteed lower bound on probabilistic safety. Case studies on linear and nonlinear stochastic systems validate the effectiveness of our approach.
\end{abstract}

\input{IEEE/Sections/1_Introduction}
\input{IEEE/Sections/2_Problem}
\input{IEEE/Sections/3_Barrier}

\input{IEEE/Sections/4_Control}
\input{IEEE/Sections/5_Theorems}
\input{IEEE/Sections/6_Experiments}

\input{IEEE/Sections/7_Conclusion}

\bibliographystyle{IEEEtran}
\bibliography{cite}

\end{document}

%% file: IEEE/Inputs/packages.tex
\IEEEoverridecommandlockouts

\usepackage{cite}
\usepackage{amsthm}  
\usepackage{amsmath,amssymb,amsfonts}
\usepackage{graphicx}
\usepackage{textcomp}
\usepackage[dvipsnames]{xcolor}         % colors
\def\BibTeX{{\rm B\kern-.05em{\sc i\kern-.025em b}\kern-.08em
    T\kern-.1667em\lower.7ex\hbox{E}\kern-.125emX}}
\usepackage[hidelinks]{hyperref}       % hyperlinks

\usepackage[utf8]{inputenc} % allow utf-8 input
\usepackage[T1]{fontenc}    % use 8-bit T1 fonts
\usepackage[acronym]{glossaries}
\usepackage{arydshln}

\usepackage[american]{babel}

\usepackage{todonotes}
\usepackage{xspace}
\usepackage[linesnumbered,ruled,noend]{algorithm2e}

\usepackage{mathtools}

\usepackage{cite}
\usepackage{amsmath,amssymb,amsfonts}
\usepackage{bbm}

\newenvironment{proof}{\noindent \textit{Proof.}}{\hfill$\square$}

\makeatletter
\newcommand{\removelatexerror}{\let\@latex@error\@gobble}
\makeatother

\usepackage{algpseudocode}
\usepackage{subcaption}
\usepackage{booktabs}
\usepackage{eso-pic}

%% file: IEEE/Inputs/commands.tex
\newcommand{\unsafe}{\mathrm{u}}
\newcommand{\safe}{\mathrm{s}}

\newcommand{\E}{\mathbb{E}}

\newcommand{\pw}{\mathbf{w}}
\newcommand{\px}{\mathbf{x}}
\newcommand{\pu}{\mathbf{u}}

\newcommand{\pX}{\mathbf{x}}

\newcommand{\erf }{\text{erf}}
 
\newcommand{\reals}{\mathbb{R}}
\newcommand{\naturals}{\mathbb{N}}
\newcommand{\N}{\mathcal{N}}
\newcommand{\naturalszero}{\mathbb{N}_{\geq 0}}
\renewcommand{\P}{\mathcal{P}}
\renewcommand{\L}{\mathcal{L}}

\newcommand{\C}{\mathbbmss{O}}
\newcommand{\NK}{\mathbb{N}_{\leq K}}
\renewcommand{\O}{\mathcal{O}}

\newcommand{\up}[1]{\overline{#1}}
\newcommand{\low}[1]{\underline{#1}}

%% file: IEEE/Inputs/theories.tex
\newtheorem{lemma}{Lemma}
\newtheorem{definition}{Definition}
\newtheorem{problem}{Problem}

\newtheorem{proposition}{Proposition}
\newtheorem{corollary}{Corollary}
\newtheorem{remark}{Remark}
\newtheorem{theorem}{Theorem}

%% file: IEEE/Inputs/acronyms.tex
\newacronym{lp}{LP}{Linear Programming}
\newacronym{gd}{GD}{Gradient Descent}
\newacronym{pwc}{PWC}{Piecewise Constant}
\newacronym{pwa}{PWA}{Piecewise Affine}
\newacronym{sos}{SOS}{Sum-of-Squares}
\newacronym{sdp}{SDP}{Semi-Definite Programming}
\newacronym{cbfs}{CBFs}{Control Barrier Functions}
\newacronym{sbfs}{SBFs}{Stochastic Barrier Functions}
\newacronym{pscbf}{PS-CBFs}{Piecewise Stochastic Control Barrier Functions}
\newacronym{cpscbf}{cPS-CBFs}{Constant Piecewise Stochastic Control Barrier Functions}
\newacronym{ibp}{IBP}{Interval Bound Propagation}

%% file: IEEE/Sections/1_Introduction.tex
\section{Introduction}

Real-world controlled systems often exhibit intricate dynamics and inherent uncertainty, making safety assurance a critical challenge in \emph{safety-critical} applications, such as autonomous vehicles \cite{shalev2017formal}, robotic systems \cite{guiochet2017safety}, and aerospace applications \cite{mazouz2021dynamics, reed2024shielded}. Traditional control design methods for such systems are often time-consuming and error-prone, with safety analysis relying heavily on extensive testing \cite{kress2018synthesis, elimelech2024falsification}.
In contrast, formal approaches such as \emph{\gls{cbfs}}~\cite{ames2019control,ames2016control}
provide a principled framework for synthesizing \emph{safe-by-design} controllers \cite{tabuada2009verification,lahijanian2015formal}.  
However, extending CBF synthesis to nonlinear systems with stochastic disturbances remains a significant challenge. 
In this work, we aim to address this gap by introducing a method for synthesizing stochastic CBFs for general discrete-time nonlinear stochastic systems. 

A CBF is a Lyapunov-like function used to enforce safety constraints by ensuring the \emph{forward invariance} of a defined safe set. When the CBF conditions are satisfied, a control law can be synthesized (extracted) to keep the system within a safe region. \gls{cbfs} have been extensively studied for deterministic systems \cite{ames2019control, ames2016control}.  
For stochastic systems, \emph{\gls{sbfs}} provide a formal framework for probabilistic safety verification under a given control policy~\cite{ mazouz2024piecewise}.  
The synthesis of CBFs using SBF certificates is an emerging research area \cite{SANTOYO2021109439, jagtap2020formal, mazouz2022safety}.  
In \cite{SANTOYO2021109439}, a safe feedback law is obtained through an iterative synthesis of an SBF; however, this approach is limited to polynomial systems and lacks guarantees on termination.  
In \cite{jagtap2020formal}, an \gls{sos} approach is employed to identify a safe control policy for generalized polynomial dynamics under uncertainty.  
In \cite{mazouz2022safety}, the \gls{sos}-convex property of the synthesized certificate is leveraged to derive a safe controller for general nonlinear systems.  
However, ensuring the SOS-convex property, even with \gls{sdp} solvers, is not always guaranteed \cite{ahmadi2009sum}.  

In this work, we propose a stochastic CBF (s-CBF) synthesis method for discrete-time nonlinear stochastic systems with continuous control sets to ensure safety over finite or infinite time horizons.  
Unlike traditional CBF and SBF methods \cite{SANTOYO2021109439, jagtap2020formal, mazouz2022safety}, which typically assume that either the barrier function or the controller is predefined, our approach jointly synthesizes both.  
However, this introduces a key challenge: the functional optimization problem requires simultaneously searching for both the barrier function (SBF) and a safe controller (CBF),  
which generally leads to a nonconvex bi-level optimization problem.
To address this complexity, we leverage \emph{piecewise} SBFs~\cite{mazouz2024piecewise}, constraining the functional optimization to a structured class of piecewise functions.
In particular, we show that by focusing on piecewise constant (PWC) functions, the synthesis of s-CBFs can be formulated as a minimax optimization problem.  
Notably, we demonstrate that this minimax problem can be efficiently solved via linear programming (LP), enabling the use of a wide variety of LP solvers and enhancing scalability.  
Our evaluations on both linear and nonlinear stochastic systems with varying dimensionalities validate the theoretical guarantees and illustrate the method's effectiveness for complex systems.  

In short, this work makes the following contributions:
\begin{itemize}
    \item A derivation of  PWC s-CBFs, formulated as a minimax optimization problem,
    \item An exact (and iteration-free) approach to solving the PWC s-CBF synthesis (minimax optimization) problem, using a linear program,
    \item Validation of the theory and illustrating the performance of the proposed method through a series of benchmark problems on linear and nonlinear stochastic systems.
\end{itemize}

\subsection*{Related Work}
\indent Probabilistic safe control design is a rich field that has been extensively researched from various lenses.
In \cite{cosner2023robust}, the problem of safe control synthesis using CBFs is extended to stochastic systems and formulated as a quadratic program.
In that work, the \gls{cbfs} are designed manually, making the solution optimal for a given barrier. 
The same holds for the method presented in \cite{vahs2024non}, which takes a belief-space approach to synthesize safe controllers for continuous-time affine-in-control SDEs under a predefined level set.
Predefining the barrier or policy a priori is a sub-optimal approach.
In this work, we simultaneously synthesize a safe controller and a barrier certificate for systems with general discrete-time nonlinear dynamics and unbounded-support additive noise.

In \cite{cauchi2019efficiency}, a method for safe control synthesis is proposed using Interval Markov Decision Processes under finite actions via dynamic programming.
That work is extended to continuous action spaces in \cite{delimpaltadakis2023interval} under certain convexity conditions on the control space.
Other methods include data-driven and learning based approaches that identify regions of the control space that ensure safe behavior, referred to as permissible strategies \cite{mazouz2024data} and shields \cite{reed2024learning}.
These approaches do not directly synthesize a safe controller that maximizes the probability of safety.
In \cite{gracia2024temporal}, a distributionally robust control synthesis method is proposed for switched stochastic systems with uncertain distributions.
Learning-based approaches have also been introduced for safe reach-avoid specifications \cite{vzikelic2023learning}, safe reinforcement learning, and learning safe controllers using Gaussian Processes \cite{berkenkamp2016safe}.
None of these works simultaneously synthesize a barrier certificate and a safe controller as a single LP optimization.

%% file: IEEE/Sections/2_Problem.tex
\section{Problem Formulation}
We consider a controlled discrete-time stochastic system described by the following stochastic difference equation
\begin{equation}
    \begin{split}
    \label{eq:system}
      &  \px_{k+1} = f(\px_k, \pu_k) + \pw_k,
    \end{split}
\end{equation}
where $\px_k \in \reals^n$ and $\pu_k \in U \subset \mathbb{R}^m$ are  respectively the state and control {input} at time $k \in \naturals$, and $f:\reals^n \times U \to \reals^n$ is a possibly nonlinear function that is assumed to be continuous. 
Term $\pw_k $ is random variable that is independent and identically distributed at every time-step with values in $\mathbb{R}^w$, which we assume to be Gaussian with zero-mean and covariance matrix $\Sigma$, i.e., 
$\pw_k \sim p_{\mathbf{w}} = \mathcal{N}(0,\Sigma)$.
The control input $\pu_k$ at state $x$ is determined by a stationary feedback controller $\pi:\mathbb{R}^n \to U$, i.e., $\pu_k = \pi(x)$. Intuitively, System \eqref{eq:system} represents a general model of nonlinear dynamical systems with additive noise, whose one-step dynamics are represented by vector field $f$.

For a given state $x \in \mathbb{R}^n,$ action $u\in U$ and (Borel measurable) set $X\subseteq \mathbb{R}^n$, we define the transition kernel $T(X\mid x, u)$ for System \eqref{eq:system} as
\begin{equation}
\label{eq:transition_kernel}
    T(X\mid x,u):= \int_{\mathbb{R}^{\mathrm{w}}} \mathbf{1}_X (f(x,u) + w )p_{\mathbf{w}}(w)dw,
\end{equation}
where $\mathbf{1}_X$ is the indicator function for $X$ such that
$\mathbf{1}_X(x) = 1$ if $x \in X$ and $\mathbf{1}_X(x) = 0$ if $x \not\in X$. For a given a controller $\pi$, the  kernel induces a  unique probability measure $\Pr$ for System \eqref{eq:system}~\cite{bertsekas2004stochastic} such that for any $N\in \naturalszero$, initial condition $x_0 \in \mathbb{R}^n$ and  measurable sets $X_0, X_k \subseteq X$, it holds that
\begin{equation}
    \label{eq:prob measure}
    \begin{aligned}
        &\Pr[\pX_0\in X_0] =  \mathbf{1}_{X_0}(x_0), \\
        &\Pr[\pX_k\in X_k   \mid \pX_{k-1}=x, \, \pi] =  T(X_k \mid x,\pi(x)).
    \end{aligned}
\end{equation}

We now define Probabilistic Safety, which is the main objective we want to optimize in this paper, through the synthesis of the controller $\pi$.
\begin{definition}[Probabilistic Safety]
    \label{def:prob safety}
    Given a bounded safe set $X_\safe \subset \reals^n$ and initial set $X_0 \subseteq X_\safe$, \emph{probabilistic safety} of System~\eqref{eq:system} under controller $\pi$ for $N$ steps is  defined as
    \begin{multline*}
        P_\safe(X_\safe, X_0, N, \pi) = \\
        \inf_{x_0 \in X_0} \Pr[\px_k \in X_\safe \; \forall k \leq N \mid \px_0 = x_0, \pi].
    \end{multline*}
\end{definition}
We now state our goal in this paper: synthesizing a controller that maximizes probabilistic safety for System \eqref{eq:system}.
\begin{problem}[Safe Control Synthesis]
    \label{Prob:synthesis}
    {Consider stochastic System~\eqref{eq:system}}, compact safe set $X_\safe \subset \reals^n$, and initial set $X_0 \subseteq X_\safe$. Given a time horizon $N \in \mathbb{N}_0 \cup \{\infty\}$,
    synthesize a controller $\pi^*$ such that
\begin{equation*}
  \pi^* \in \arg\max_{\pi}  P_\safe(X_\safe, X_0, N, \pi^*).
\end{equation*}
\end{problem}
Problem~\ref{Prob:synthesis} seeks an optimal feedback controller that maximizes Probabilistic Safety for System \eqref{eq:system}. Because of the nonlinearity of $f$, despite the additivity of the noise, this problem presents several challenges that need to be addressed. First of all, $P_\safe$ cannot be computed in closed form. Furthermore, in Problem \ref{Prob:synthesis}, we seek an optimal feedback controller with $U$ possibly containing uncountable many controls. 
Consequently, existing approaches that generally consider a finite number of actions in $U$ cannot be applied in this setting \cite{cauchi2019efficiency}. 

\begin{remark}  
    \label{remark:not limited to gaussian}
    While we assume \(\pw_k\) is Gaussian, this assumption is not strictly required by our framework and is made solely for clarity of presentation. In fact, the only requirement on \(\pw_k\) is that the resulting transition kernel satisfies \(T(X \mid x, u)\) being computable (or analytically over-approximable). This condition holds for Gaussian distributions \cite{adams2022formal}. 
\end{remark}

%% file: IEEE/Sections/3_Barrier.tex
\section{Stochastic Control Barrier Functions}
\label{sec:stochastic-cbf}
In this section, we provide an overview of stochastic CBFs (s-CBFs), and their relation to safe control synthesis. Consider System~\eqref{eq:system} with safe set $X_\safe \subset \reals^{n}$, initial set $X_0 \subseteq X_\safe$, and unsafe set $X_\unsafe = \reals^{n}\setminus X_\safe$. 
The following definition provides the conditions for an s-CBF.

\begin{definition}[Stochastic CBF (s-CBF)]
\label{def:barrier}
A real-valued function $B:\mathbb{R}^n \to \mathbb{R}_{\geq 0}$
is called a \emph{stochastic control barrier function} (s-CBF) for System~\eqref{eq:system} if there exist controller $\pi$ and scalars $\eta, \beta \in [0,1]$ such that
\begin{subequations}
    \begin{align}
       \label{eq:barrier_nonnegative} & B(x) \geq 0 &&\forall x\in \mathbb{R}^n \\
        \label{eq:barrier_unsafe} &B(x) \geq 1  &&\forall x\in X_\unsafe \\
        \label{eq:barrier_initial}  &B(x) \leq \eta  &&\forall x\in X_0 \\
       \label{eq:control_martingale}  &\E[B(f(x, \pi(x)) + \pw | x, \pi(x)] \! \leq \!  B(x) + \beta  && \forall x\in X_\safe.
    \end{align}
\end{subequations}
\end{definition}

Let $B$ with scalar $\eta,\beta$ be a s-CBF for System~\eqref{eq:system} under $\pi$.  Then, it holds that, for time horizon $N\in \naturalszero$ \cite{laurenti2023unifying}, 
    \begin{align}
        {\inf_{x_0 \in X_0}} P_\safe(X_\safe, {x_0},N, \pi)
        \geq 1-( \eta + \beta N). \label{eq:probability-sCBF}
    \end{align}
In other words, a valid s-CBF provides a probabilistic lower bound on remaining in the safe set via \eqref{eq:probability-sCBF}. 
If the right hand side of the expression is greater or equal to a pre-defined safety probability threshold $\delta_\safe \in [0,1]$, i.e., $1 - (\eta + \beta N) \geq \delta_\safe$, then $B$ is called a barrier \emph{certificate}. 

Finding such a $B$ requires solving a functional optimization problem. 
This problem is particularly challenging due to the simultaneous search for both a function $B$ and a controller $\pi$ that maximizes safety. Unlike standard verification and control synthesis problems, where either $B$ or $\pi$ are fixed, respectively, here both must be optimized. 
Moreover, the martingale condition in \eqref{eq:control_martingale} introduces additional difficulties as it requires the expectation of the composition of $B$ with the dynamics of System~\eqref{eq:system}. 
To mitigate these challenges, we focus on piecewise constant (PWC) s-CBFs.

\subsection{Piecewise Constant s-CBFs}
Piecewise constant (PWC) functions provide an expressive class of functions for barrier synthesis while addressing some of the challenges in the functional optimization problem discussed above~\cite{mazouz2024piecewise}.

Consider a partition $\{ X_1,\ldots,X_K \}$ of safe set $X_\safe$ in $K \in \mathbb{N}$ compact sets, i.e.,
\begin{equation}
    \label{eq:partitions}
    \cup_{i=1}^K X_i = X_\safe \quad \text{ and } \quad X_i\cap X_j = \emptyset \quad \forall i\neq j \in \NK,
\end{equation}
such that vector field $f$ is continuous in each region $X_i$. 
Note that the boundary of each set has measure zero w.r.t. $T(\cdot  \mid x,\pi)$ for any $x \in X_\safe$. Further, let $B$ be a PWC function such that,
for all $i \in \NK$, $b_i \in \reals_{\ge 0}$ and
\begin{equation}
    \label{eq:pwf}
    B(x) = 
    \begin{cases}
        b_i & \text{if } x \in X_i\\
        1 & \text{otherwise.}
    \end{cases}
\end{equation}

\noindent
Then, the following corollary characterizes the PWC s-CBF.

\begin{corollary}[PWC s-CBF]
    \label{col:ps-cbf}
    The PWC function in~\eqref{eq:pwf} is an s-CBF for System~\eqref{eq:system} if, for every $i \in \NK$, there exist control $u_i \in U$ and scalars $\eta, \beta_i \in [0,1]$ such that
    \begin{subequations} 
        \begin{align}
            & b_i\geq 0,  && \forall x \in X_i, \label{eq:piecewise_nonnegative}\\
            & b_{i} \leq \eta, && \forall x \in X_i \cap X_0, \label{eq:piecewise_initial}\\
            & \sum_{j = 1}^K b_j T(X_j \mid x, {u_i})  && \nonumber \\
            & \qquad  + T(X_\unsafe \mid x, {u_i}) \leq  b_i + \beta_{i} && \forall x \in X_i, \label{eq:piecewise_expectation}
        \end{align}
    \end{subequations}
    where 
    $T$ is the transition kernel in \eqref{eq:transition_kernel}. Then, for a controller $\pi$ that assigns $\pi(x) = u_i$ for all $x \in X_i$ and for all $i \in \NK$, it follows that
    \begin{equation}
        \label{eq:pwb_p-safe}
        P_\safe(X_\safe,X_0, N, \pi) \geq 1 - (\eta + N\cdot \max_{i \in \NK }\beta_i). 
    \end{equation}
\end{corollary}

The proof follows directly from~\cite[Theorem 2]{mazouz2024piecewise}.
Following this corollary, the synthesis of PWC s-CBF can be formulated as an optimization problem, where the size of the partition is a design parameter. 
The benefit of the formulation is that the expectation in \eqref{eq:control_martingale} is reduced to a simple summation in~\eqref{eq:piecewise_expectation}.  
The challenge however lies in the characterization of transition kernel $T$, which is parameterized with $u_i$.

%% file: IEEE/Sections/4_Control.tex
\section{PWC s-CBF Synthesis}
\label{sec:cbf-synthesis}

In this section, we formally set up the general optimization problem for  synthesis of PWC s-CBF.
Throughout this section, we consider a specific partition of $X_s$ as detailed below.

For covariance $\Sigma$ of noise $\pw_k$, let $\mathcal{T}$ denote the Mahalanobis (linear) transformation\footnote{Specifically, $\mathcal{T}= \Gamma^{1/2} V^T$, where $\Gamma = V^T \Sigma V$ is a diagonal matrix with the eigenvalues of $\Sigma$ as the entries, and $V$ is the corresponding orthonormal eigenvector matrix.}. 
Denote by $\mathcal{T}(X_s)$, the transformation of $X_s$ by $\mathcal{T}$.  Then, we consider the partition in~\eqref{eq:partitions} to be an axis-aligned grid in $\mathcal{T}(X_s)$, i.e., for every $i \in \NK$, $\mathcal{T}(X_i)$ is a hyper-rectangle.
Using this partition, we derive an exact expression for the transition kernel.

\begin{proposition}[{\cite[Proposition 1]{cauchi2019efficiency}}]
    \label{prop:tran kernel}
    Let $X_j \subset \reals^n$ be a region such that $\mathcal{T}(X_j) = [\low{v}^{(1)}, \bar{v}^{(1)}] \times \ldots \times [\low{v}^{(n)}, \bar{v}^{(n)}]$ is a hyper-rectangle, where $\low{v}^{(i)}, \bar{v}^{(i)} \in \reals$ denote the lower and upper bounds of the $i$-th dimension of the hyper-rectangle.
    Then, for a given point $x \in \reals^n$ and control $u\in U$, 
    \begin{equation}
    \label{eq:proderf}
      T(X_j | x, u) =
     \frac{1}{2^n} \prod_{i=1}^{n} \left (
        \erf (\frac{y^{(i)} - \low{v}^{(i)}}{\sqrt{2}}) 
        -
        \erf (\frac{y^{(i)} - \bar{v}^{(i)}}{\sqrt{2}}) 
        \right )\!,
    \end{equation}
    where $\erf(\cdot)$ is the error function, $y = \mathcal{T} f(x, u)$, and $y^{(i)}$ is the $i$-th component of $y$. 
\end{proposition}

Proposition~\ref{prop:tran kernel} provides an analytical expression for $ T $ given $ x $ and $ u $, facilitating the formulation of the s-CBF constraint in \eqref{eq:piecewise_expectation}.  
Nevertheless, note that the constraint in \eqref{eq:piecewise_expectation} must hold for all $ x \in X_i $.  
To simplify its formulation, we aim to derive affine lower and upper bounds on \eqref{eq:proderf} for all $ x \in X_i $.
The following lemma allows us to obtain these bounds using \gls{ibp} techniques \cite{zhang2018efficient, Mathiesen2022, xu2020automatic}.

\begin{lemma}[Affine Transition Kernel Bounds]
\label{lemma:affine tran kernel bounds}
Consider partition regions $X_i,X_j$ and the transition kernel in~\eqref{eq:proderf} for $x \in X_i$ and $u \in U$.  Then, using \gls{ibp},  affine functions 
\begin{subequations}
    \begin{align}
        \low{T}_{ij}(x,{u}) & = A_{ij}^{\bot}[x, {u}_i]^{T} + c_{ij}^{\bot}, \\
        \up{T}_{ij}(x,{u}) & = A_{ij}^{\top}[x, {u}_i]^{T} + c_{ij}^{\top},   
    \end{align}
\end{subequations}
    where $A_{ij}^{\top},A_{ij}^{\bot} \in \reals^{1\times (n+m)}$ are vectors and
$c_{ij}^{\top}, c_{ij}^{\bot} \in \reals$ are scalars, can be computed such that
\begin{align}
    \label{eq:sound_bounds}
    \low{T}_{ij}(x,{u}) \leq T(X_j \mid x, u) \leq \up{T}_{ij}(x, {u}) 
    % \quad \forall x \in X_i, \forall u\in U.
\end{align}
for all $x \in X_i$ and for all  $u\in U$.
\end{lemma}
\begin{proof}
Using backward mode linear relaxation perturbation analysis (LiRPA), the lower and upper bounds are sound. 
The proof follows directly from \cite[Theorem 1]{xu2020automatic}. 
\end{proof}

\begin{remark} 
    Lemma~\ref{lemma:affine tran kernel bounds} can be extended to transition kernels beyond~\eqref{eq:proderf}. In fact, it can be generalized for any continuous and differentiable kernel. This implies that our approach is not limited to Gaussian noise $\pw_k$, as stated in Remark~\ref{remark:not limited to gaussian}.
\end{remark}

The benefit of Lemma~\ref{lemma:affine tran kernel bounds} is that it allows us to represent the bounds of $T$ using parameterizations in $x$ and $u$. 
However, these parameters are different in nature when it comes to PWC s-CBF synthesis.  That is, $x$ is a free parameter, i.e., s-CBF properties have to hold for all $x\in X_i$, whereas $u \in U$ is a decision variable, i.e., we seek to find the optimal $u$.
With this view and using the bounds in \eqref{eq:sound_bounds}, we can define the feasible set of the transition kernels for a given $u \in U$ and all $x \in X_i$ as
\begin{align}
    \label{eq:P-simplex}
        & \mathcal{P}_i(u) = \Big\{ T_i = (T_{i1},\ldots,T_{iK},  T_{i\unsafe})  \in [0,1]^{K+1} \quad s.t. \nonumber\\
        & \low{T}_{ij}(x,{u}) \leq T_{ij} \leq \up{T}_{ij}(x, {u})
        \quad \forall j \in \NK \cup \{\unsafe\}, \forall x \in X_i
       \nonumber \\
        & \hspace{3cm} \sum_{j=1}^{K} T_{ij} + T_{i\unsafe} = 1, 
        \Big \}.
\end{align}
% \RM{Add for all $x \in X_i$, change order.}
% \ml{notations is incorrect. fix it.} \RM{Which notations?}
This feasible transition kernel $\mathcal{P}_i(u)$ set forms a simplex, i.e., a convex polytope. 
This implies that we can rewrite the constraint in~\eqref{eq:piecewise_expectation} as a set of linear constraints.
Based on this, the following theorem defines the general PWC s-CBF optimization problem.

\begin{theorem}[PWC s-CBF Synthesis]
    \label{th:ps-cbf-general}
    Given a $K$-partition of $X_\safe$, let $\mathcal{B}_K$ be the set of PWC functions of the form in~\eqref{eq:pwf}. Further, let vector $\vec{u} = (u_1, \ldots, u_K) \in U^{K}$ and $\mathcal{P}(\vec{u}) = \mathcal{P}_1(u_1) \times \ldots \times \mathcal{P}_K(u_K)$, where each $\mathcal{P}_i(u)$ is the set of feasible transition kernels in~\eqref{eq:P-simplex}. 
    Finally, let $B^*$ and $\vec{u}^*$ be the solutions to the following optimization problem
    \begin{equation}
        (B^*,\vec{u}^*) = \arg\min_{B \in \mathcal{B}_K, {\vec{u} \in {U}^{K}}} \;  {\max_{(T_i)_{i=1}^K \in \mathcal{P}(\vec{u})}} \; \eta + N \beta 
    \end{equation}
    subject to 
    \begin{subequations} 
        \begin{align}
           \label{eq:ps-cbf-nonnegative} & b_i \geq 0, && \forall i \in \NK,\\
           \label{eq:ps-cbf-initial}  & b_i \leq \eta, &&\forall i : X_i \cap X_0 \neq \emptyset,\\
           \label{eq:ps-cbf-expectation}  & \sum_{j = 1}^{K} b_j \cdot T_{ij} + T_{i\unsafe} \leq  b_i + \beta_{i} && \forall i\in \NK, \\
            & T_{ij} \in [\low{T}_{ij}(x, u_{i}), \up{T}_{ij}(x, u_{i})] &&  \forall x \in X_i, \forall i \in \NK, \\
             &0 \leq \beta_{i} \leq \beta && \forall i\in \NK.
        \end{align}
    \end{subequations}
    Then, $B^*$ is a PWC s-CBF along with $\vec{u}^*$ that maximize safety probability, i.e., the RHS of~\eqref{eq:pwb_p-safe}.
\end{theorem}
\begin{proof}
    We begin by observing that if $\forall i, b_i$s satisfy conditions~\eqref{eq:ps-cbf-nonnegative}-\eqref{eq:ps-cbf-expectation}, then conditions in~\eqref{eq:piecewise_nonnegative}-\eqref{eq:piecewise_expectation} are satisfied for the barrier function in~\eqref{eq:pwf}. The optimization seeks to maximize the safety probability for feasible values of the transition kernel $\mathcal{P}_i$, formulated as a minimax problem. $\mathcal{P}_i$ as defined in~\eqref{eq:P-simplex} is a function of the controller $\vec{u}$, and the optimization aims to find a controller that minimizes $(\eta + N\beta)$.    
    The values for $\beta_i$ are bounded in accordance with constraint~\eqref{eq:ps-cbf-expectation}. Then, following Corollary \ref{col:ps-cbf},
    $\mathbb{E}[B(f(x, u)) + \pw \; | \; {x, u}] - B(x)
    = \sum_{j = 1}^K b_j T(X_j \mid x, {u_i})  + T(X_\unsafe \mid x, {u_i}) - b_i \leq  \beta_{i} $.
    Taking the maximum over all $\beta_i$s is an upper bound on the martingale condition. Hence, optimizing over objective $(B^*,\vec{u}^*)$ finds a controller that maximizes the probability of safety defined in~\eqref{eq:probability-sCBF}.
\end{proof}

\begin{remark} In Theorem \ref{th:ps-cbf-general}, $N\in \naturalszero$ is finite. However, by setting $\beta_i = 0$ for all $i\in \NK$, the theorem can easily be extended to infinite-horizon problems.
\end{remark}

By solving the optimization problem in Theorem~\ref{th:ps-cbf-general}, we obtain a PWC controller as defined in Corollary~\ref{col:ps-cbf}, i.e., $\pi(x) = u^*_i$ for all $x \in X_i$. This controller is optimal within the class of PWC controllers for the given $K$-partition.  
In fact, this controller approximates the true optimal stationary controller $\pi^*$ as $K \to \infty$. Specifically, per \cite[Proposition 1]{mazouz2024piecewise}, as the partition size increases, the PWC SBF converges to a safety probability that is at most the optimal safety probability achievable with continuous SBFs. Consequently, if an optimal continuous s-CBF exists that maximizes $P_\safe$, the same $P_\safe$ can be attained with a sufficiently large $K$.  

Thus, solving the optimization problem in Theorem~\ref{th:ps-cbf-general} effectively solves Problem~\ref{Prob:synthesis} as $K \to \infty$. In the next section, we present an exact method for solving the problem in Theorem~\ref{th:ps-cbf-general} using dual programming.

%% file: IEEE/Sections/5_Theorems.tex
\section{Dual Linear Program}
\label{sec:dual}

The optimization problem in Theorem \ref{th:ps-cbf-general} is a minimax problem, with the following decision variables:
\begin{subequations}
    \begin{align}
        &b = (b_1,\ldots,b_K) \in \mathbb{R}^{K}_{\geq 0}, && \label{eq:b-variable}\\
        & \vec{u} = (u_1,\ldots,u_K) \in U^{K}, && \label{eq:u-variable}\\
        &T_i = (T_{i1}, \ldots, T_{iK}, T_{i\unsafe}) \in \mathcal{P}_i(u_i) && \forall i\in \NK, \label{eq:p-variable}\\
        & \beta_i \in \mathbb{R}_{\geq 0} && \forall i\in \NK
        \label{eq:beta-variable},\\
        & \eta,\beta \in \mathbb{R}_{\geq 0}. &&
    \end{align}
\end{subequations}

The difficulty in the computability of this optimization problem is twofold. First, the product of decision variables $b_j$ and $T_{ij}$ makes the optimization problem bilinear. Second, the transition kernels $T_i$ are dependent on the control $u_i$, which makes this an embedded optimization problem. Generally, this class of optimization problems is non-convex, thus convex solvers which provide efficiency and convergence guarantees cannot be utilized. In this section, we propose a lossless convexification of the problem based on dual linear programming, such that the optimal solution to the problem can be efficiently computed. 

We start by observing that $x$ is an independent variable, and $u_i$ a decision variable. 
Since $x \in X_i$ is free, for every $i \in \NK$, we can set $x = x_i' - x_i''$, where decision variables $x_i', x_i'' \geq 0$. We conveniently define the vector $\tilde{z}_i = [x_i' - x_i'', u_i]^{T} \in \reals^{(n+m)\times 1}$. Then, the 
feasible transition kernel set $\mathcal{P}_i(u_i)$, which is a simplex, can be represented as
{\small 
\begin{align}
    \nonumber
    \mathcal{P}_i(u_i)  =&  \{T_i : H_i \, 
    \begin{bmatrix}
        T_i \\ \tilde{z_i}
    \end{bmatrix}
    \leq h_i \}, \\
    \label{eq:simplex-P}
    :=&  \underbrace{
    \begin{bmatrix}
        -I_{K+1} & A_i^{\bot} \\
        {-\vec{1}} & {\vec{0}} \\
        \vec{0}_{n + m} & -\vec{1}_{n + m} \vspace{2mm} \\
          \hdashline \vspace{-2mm} \\
        I_{K+1} & - A_i^{\top} \\
        {\vec{1}} & {\vec{0}} \\
         \vec{0}_{n + m} & \vec{1}_{n + m} 
    \end{bmatrix}
    }_{H_i}
    \begin{bmatrix}
       T_i \\~ \tilde{z_i}
    \end{bmatrix}
    \leq 
    \underbrace{
    \begin{bmatrix}
        -c_i^{\bot} \\
        {-1} \\
        -\tilde{z}_i^{\bot} \vspace{2mm} \\
            \hdashline \vspace{-2mm} \\
        c_i^{\top} \\
        {1} \\
        \tilde{z}_i^{\top}
    \end{bmatrix}}
    _{h_i},
\end{align}}
\noindent
where matrix {$H_i \in \reals^{2(K + n + m + 2) \times (K+n+m+1)}$}, vector $h_i \in \reals^{2(K+ n + m +1) \times 1}$, 
$I_{K+1}$ is the identity matrix of size $(K+1) \times (K+1)$, and $\vec{0}$ and $\vec{1}$ are vectors of zeros and ones of appropriate dimensions, respectively.
The inequality in~\eqref{eq:simplex-P} is applied element-wise to vectors. The lower and upper bounds on variable $\tilde{z}_i$ are denoted by  $ \tilde{z}_i^{\bot}, \tilde{z}_i^{\top} \in \reals^{(n+m)\times1}$, respectively. Further, 
{\small 
\begin{align*}
    A_i^{\bot} & = \begin{bmatrix} A_{i1}^{\bot}, \hdots A_{iK}^{\bot},& A_{iu}^{\bot} \end{bmatrix}^{T} \quad &&\in \reals^{(K+1) \times (n+m)}, \\
    A_i^{\top} & = \begin{bmatrix} A_{i1}^{\top}, \hdots A_{iK}^{\top},& A_{iu}^{\top} \end{bmatrix}^{T}  \quad &&\in \reals^{(K+1) \times (n+m)}, \\
    c_i^{\bot} & = \begin{bmatrix} c_{i1}^{\bot}, \hdots c_{iK}^{\bot},& c_{iu}^{\bot} \end{bmatrix}^{T} \quad &&\in \reals^{(K+1) \times 1}, \\
    c_i^{\top} & = \begin{bmatrix} c_{i1}^{\top}, \hdots c_{iK}^{\top},& c_{iu}^{\top} \end{bmatrix}^{T} \quad &&\in \reals^{(K+1) \times 1}.
\end{align*}
}
To constrain the sum of $T_{ij}$ to $1$ per~\eqref{eq:P-simplex}, two inequality constraints are added. For convenience, the rows of the matrix $H_i$ and vector $h_i$ can be rearranged
{\small 
\begin{align*}
% \hspace{-2mm}
    \tilde{H}_i = \!
    \begin{bmatrix}
        -I_{K+1} & A_i^{\bot} \\
        I_{K+1} & - A_i^{\top} \\
        -\vec{1} & {\vec{0}} \\
       \vec{1} & {\vec{0}} \vspace{1mm} \\
          \hdashline \vspace{-3mm} \\
        \bar{0}_{n + m} & -\bar{1}_{n + m} \\
        \bar{0}_{n + m} & \bar{1}_{n + m} 
    \end{bmatrix} \!=\!
    \begin{bmatrix}
        \tilde{H}_i^{p}\vspace{1mm} \\
          \hdashline \vspace{-3mm} \\
        \tilde{H}_i^{z}
    \end{bmatrix}
    , \;
    \tilde{h_i} = \!
        \begin{bmatrix}
        -c_i^{\bot} \\
        c_i^{\top} \\
         {-1} \\
         {1} \vspace{1mm} \\
          \hdashline \vspace{-3mm} \\
        -\tilde{z}_i^{\bot} \\
        \tilde{z}_i^{\top}
    \end{bmatrix}
    \! = \!
    \begin{bmatrix}
        \tilde{h}_i^{p}\vspace{1mm} \\
          \hdashline \vspace{-3mm} \\
        \tilde{h}_i^{z}
    \end{bmatrix}.
\end{align*}
}
Using $\tilde{H}_i$ and $\tilde{h}_i$, the inequality can be decomposed into two separate inequalities. It is noted that inequality $\tilde{H}_i^{z} \tilde{z}_i \leq  \tilde{h}_i^{z}$ simply denotes the bounds on decision variables $\tilde{z}_i$. 
The second inequality establishes the bounds on the transition kernels $T_{i}$, parameterized by $\tilde{z}_i$, and can be written as
{\small 
\begin{align}
\label{eq:bounds_T}
 \underbrace{
     \begin{bmatrix}
     \tilde{H}_i^{p_1} &
          \tilde{H}_i^{p_2} \\
    \end{bmatrix} }
    _{\tilde{H}_i^{p}}
       \begin{bmatrix}
        T_i \\ \tilde{z_i}
    \end{bmatrix}
    \leq 
    \underbrace{
    \begin{bmatrix}
        -c_i^{\bot} \\
        c_i^{\top} 
    \end{bmatrix}
    }_{\tilde{h}_i^{p}}, 
\end{align}
}
where
{\small 
\begin{align*}
    \tilde{H}_i^{p_1} & = [-I_{K+1},\; I_{K+1}; \; -\vec{1}; \; \vec{1}]^{T} && \in \reals^{2(K+2)\times(K+1)} \\
    \tilde{H}_i^{p_2} & = [A_i^{\bot},\; -A_i^{\top};  \; \vec{0}; \; \vec{0}]^{T}&& \in \reals^{2(K+2)\times(n+m)}
\end{align*}
}

Next, define vector $\bar{b} = (b,1)$, where $b$ is as defined in~\eqref{eq:b-variable}, and dual variable $\lambda_i \in \mathbb{R}^{2(K+2)}_{\geq 0 }$. Using duality, the problem in Theorem~\ref{th:ps-cbf-general} can be formalized as an LP.

\begin{theorem}[PWC s-CBF Dual LP]
 \label{th:dual}

The optimal solution $(b^{*}, u^{*}, \beta^{*}, \eta^{*})$ to the LP formulated below is equivalent (i.e., zero duality gap) to the optimal solution of the optimization problem in Theorem~\ref{th:ps-cbf-general}. 
\begin{equation*}
    \min_{b, t_i, \lambda_i} \eta + \beta N
\end{equation*}
subject to
\begin{subequations}
    \begin{align}
        & 0 \leq b_{i} \leq 1 && \forall i\in \NK,\\ 
        & b_{i} \leq \eta  &&\forall i : X_i \cap X_0 \neq \emptyset, \\   
        &  (\tilde{h}_i^{p})^{T} \lambda_i \leq {t_i} && \forall i\in \NK, \label{eq:dual-c} \\
        & (\tilde{H}_i^{p_1})^{T} \lambda_i = \bar{b}&& \forall i\in \NK,\\
        & {(\tilde{H}_i^{p_2})^{T} \lambda_i = 0} && \forall i\in \NK,\\
        &  \tilde{z}_i^{\bot} \leq \tilde{z}_i \leq  \tilde{z}_i^{\top} && \forall i\in \NK,\\
        &  \lambda_i \geq 0 && \forall i\in \NK,\\
        &  t_i \geq 0  && \forall i\in \NK, \\
            & 0 \leq \beta_{i} \leq \beta && \forall i\in \NK. \label{eq:dual-i}
     \end{align}
 \end{subequations}
\end{theorem}

\begin{proof}
    We split the proof of this optimization into three parts. First we show how the constraint in~\eqref{eq:ps-cbf-expectation} can be formulated as an inner minimax problem \cite{boyd2004convex}. Then, using Lagrangian duality \cite{ye_lagrangian, hager1976lagrange}, we show that the inner problem can be formulated as an LP. Finally, we show that, given the convex nature of the inner problem, the KKT conditions allow us to write the entire optimization in Theorem~\ref{th:ps-cbf-general} as one LP \cite[Sec. III]{sinha2017review} \cite{dempe2016solution}.
    
    Write the constraints $\bar{b}^{T} T_i \leq b_i + \beta_i$ and condition~\eqref{eq:bounds_T} as an inner optimization problem
     {\small 
     \[
        \left(\begin{aligned}
      { \min_{u_i}}     \max_{T_i\in \P_i} & \quad \bar{b}^{T}T_i\\
            \mathrm{s.t.} & \quad \tilde{H}_i^{p}  
            \begin{bmatrix}
                T_i \\ \tilde{z_i}
            \end{bmatrix} 
            \leq \tilde{h}_i^{p}
        \end{aligned}\right) \leq b_i + \beta_i.
    \]
    }
    Let the Lagrangian associated with the maximization be
    {\small 
    \begin{align*}
        \L(T_i, \tilde{z}_i, \lambda) & = \bar{b}^{T}T_i - \lambda_{i}^{T}(\tilde{H}_i^{p}\begin{bmatrix}
                T_i \\ \tilde{z_i}
            \end{bmatrix} 
            - \tilde{h}_i^{p} ) \\
            & = \bar{b}^{T}T_i - \lambda_{i}^{T}( \tilde{H}_i^{p_1} T_i - \tilde{H}_i^{p_2} \tilde{z}_i- \tilde{h}_i^{p} )
    \end{align*}
    }
    For strong duality, the following KKT condition must hold
    $$ \nabla_{T_i, \tilde{z}_i} \L(T_i, \tilde{z}_i, \lambda)  = 0
    $$
The first component of the derivative yields
{\small 
\begin{align}
   \nonumber \frac{\partial \mathcal{L}}{\partial T_i} & = \bar{b}^{T} - \lambda_i^{T}\tilde{H}_i^{p_1} = 0, \\
    (\tilde{H}_i^{p_1})^{T} \lambda_i & = \bar{b}^{T}.
\end{align}
}
Combining the above equation with the conditions in Eq.~\eqref{eq:bounds_T} and~\eqref{eq:ps-cbf-expectation} yields
{\small 
\begin{align*}
     (\tilde{h}_i^{p})^{T} \lambda_i \leq b_i + \beta_i
\end{align*}
}
The second component of the derivative yields
{\small 
\begin{align*}
    \frac{\partial \mathcal{L}}{\partial \tilde{z}_i} & = \bar{b}^{T}\frac{\partial T_i}{\partial \tilde{z}_i}
    - \lambda_i^{T}(\tilde{H}_i^{p_1} 
    \frac{\partial T_i}{\partial \tilde{z}_i}
    - \tilde{H}_i^{p_2}) = 0 \\
    & =  \underbrace{(\bar{b}^{T} - \lambda_i^{T}\tilde{H}_i^{p_1})}_{: =0}
    \frac{\partial T_i}{\partial \tilde{z}_i} + \lambda_i^{T}\tilde{H}_i^{p_2} = 0 \\
    \lambda_i^{T}\tilde{H}_i^{p_2} & = 0 
\end{align*}
}
Introducing auxiliary variable $t_i$ to represent $\max_{T_i\in \P_i}  \bar{b}^{T}T_i$, the inner optimization can be substituted for an equivalent asymmetric dual problem \cite{boyd2004convex, mazouz2024piecewise}, yielding an LP
{\small 
\begin{align*}
    & \min_{u_i} t_i \qquad \text{subject to: } \quad \eqref{eq:dual-c}-\eqref{eq:dual-i}
 \end{align*}
}
The inner optimization problem is convex, which allows one to re-write the bi-level optimization problem as a single LP using the KKT conditions \cite{sinha2017review, dempe2016solution}.
\end{proof}

% \noindent 
\paragraph*{Computational Complexity} 
The time complexity of a standard primal LP is $\O(n^2m)$ where $n$ is the number of decision variables and $m$ is the number of constraints \cite{boyd2004convex}. 
The program in the above theorem has $n = 3K^2 + 8K + 2$ and $m = 2K^2 + 10K + L + 1$, where $L = \lvert \{i : X_i \cap X_0 \neq \emptyset \} \rvert$.

We finally note that, while traditional LP solvers can be slow for large-scale problems, the LP formulation presented here enables alternative methods like gradient descent, which can be more scalable and efficient. In future work, we plan to explore such methods to achieve even higher scalability.

%% file: IEEE/Sections/6_Experiments.tex
\section{Experiments}

\label{sec:experiments}

\input{IEEE/Sections/6a-2D}

\begin{table}[b!]
\caption{Benchmark results. In the table, $K$ denotes the size of the partition of $X_s$, and $t_\text{synth}$ is the optimization (PWC s-CBF synthesis) time. For all experiments, a total of $N=50$ time-steps is used.}
\label{table:results}
\centering
\begin{tabular}{@{}llllll@{}}
\toprule
System & $K$ & $\eta$ & $\beta$ & $P_s$ & $t_\text{synth}$ (s) \\ \midrule
2D Linear & 81 &  0.090 & 0.00420 & 0.70 & 0.91 \\
Convex & 100 & 0.002 & 0.00036 & 0.98 & 2.42 \\
 & 400 & 0.002  & 0.00034 & 0.98 & 19.85 \\ \midrule
2D Linear & 81 & 0.11 & 0.00610 & 0.59 & 1.47 \\
Non-Convex & 100 & 0.03 & 0.00022  & 0.96 & 3.19 \\
 & 400 & 0.03 & 0.00019 & 0.96 & 27.89 \\ \midrule
3D Temp. Model & 500 & 0.0475 & 0.00025 & 0.94 & 301.47 \\
 & 900 & 0.0429 & 0.00023 & 0.95 & 497.03 \\ \midrule
4D Unicycle & 1800 & 0.07 & 0.0220 & 0.82 & 2789.67 \\
 & 2400 & 0.02 & 0.0006 & 0.95 & 3945.77 \\ \bottomrule
\end{tabular}
\end{table}

We demonstrate the efficacy of our approach on various benchmarks. 
We consider a 2D linear system with a convex and nonconvex safe set, a 3D polynomial room temperature model \cite{girard2015safety}), and a 4D nonlinear unicycle model \cite{de2000stabilization}.

The safety probability bounds obtained with our framework are validated by 500 Monte
Carlo simulations, randomly initialized in $X_0$. 
Our IBP code that computes affine transition kernel bounds is in Python, and the PWC s-CBF tool is in Julia. All computations were performed on a Linux machine with 3.9 GHz
8-core CPU and 128 GB memory. 
In each of the experiments, we set a safety threshold of $\delta_s = 0.95$ (that is, 95\%) for $N= 50$ time steps.

Results for the various benchmarks considered are reported in Table~\ref{table:results}. As can be observed, with sufficient partitions for each case study, our method is able to synthesize controllers and a barrier certificate that yield $P_s \geq \delta_s = 0.95$. We discuss individual case studies in detail below.

\input{IEEE/Sections/6b-room}

\input{IEEE/Sections/6c_unicycle}

\subsection{Linear System}

We consider the following 2D linear system with dynamics
\begin{align*}
    x_{k+1} = 1.05 \, x_k + 0.1 \, u_k + \mathbf{w}
\end{align*}
where 
% $I$ is the identity matrix, and 
$\pw \sim \N(0,10^{-2}I)$. The operating domain is $X \in [-1, 1]^2$, and the initial set is $X_0 \in [0.4, 0.5]^2.$ Note that the inherent (uncontrolled) dynamics of the system are unstable. We consider two cases
\begin{enumerate}
    \item Convex safe set: $X_s = X$ (see Fig.~\ref{fig:mc1}), and
    \item Non-convex safe set
    $X_s = X \setminus \C_1(c_1,\epsilon_1)$, where for obstacle $\C_1$: 
        $c_1 = (0.15, 0.15)$ and $\epsilon_1 = (0.05, 0.05)$ (see Fig.~\ref{fig:mc1-nonconvex}).
\end{enumerate}

For the convex case, the synthesized control law and barrier are shown in Figs. \ref{fig:control} and \ref{fig:barrier}, respectively.  The probability of safety for the $N =50$ time steps is $P_s = 0.98$. 
In Fig. \ref{fig:mc1}, a Monte Carlo simulation plot is provided under the synthesized controller for 500 trajectories over a duration of $50$ time steps.  In each case, the
randomly sampled trajectories resulted in
zero violations, i.e., the system is 100\% safe. This is in line with the lower bound on $P_s$ that we obtain. 
To get to $P_s = 0.98$, the barrier and control synthesis computation takes $\approx$ 2.42 seconds (for $K=100$ partitions).

For the non-convex case, the obtained probability of safety for $N= 50$ equals $P_s = 0.96$. The synthesized control law and barrier are shown in Figs. \ref{fig:control-nonconvex} and \ref{fig:barrier-nonconvex}, respectively. Observe in Fig.~\ref{fig:barrier-nonconvex} that, different from the convex case, the barrier value equals to one at the obstacle region. Further, notice in Fig.~\ref{fig:control-nonconvex} that the controller near the obstacle properly pushes the vector field in a direction away from unsafety. 

The Monte Carlo simulations are depicted in Fig.~\ref{fig:mc1-nonconvex}. A total of 6 out of 500 trajectories violated the safety boundaries, entering the unsafe set (obstacle region). This  
indicates a frequentist probability of safety of 99\%, slightly higher than the bound we obtain. For this experiment, the synthesis computation takes about 3.19 seconds for the $K=100$ partitions, to get to $P_s \geq \delta_s$.

\subsection{Temperature Regulation}

The following room temperature regulation  model has deterministic dynamics 
\begin{align*}
  \mathrm{T}_{k+1}^i = & \mathrm{T}_k^{i} + \alpha_1(\mathrm{T}_k^{i+1} + \mathrm{T}_k^{i-1} - 2\mathrm{T}_k^{i}) +
  \alpha_2(\mathrm{T}_e - \mathrm{T}_k^{i}) \\
  + & \alpha_3(\mathrm{T}_h - \mathrm{T}_k^{i})u_i 
  + \mathbf{w},
\end{align*}
to which noise $\pw \sim \N(0, 10^{-2}I)$ is added. A total of $n=3$ rooms are considered, where $\mathrm{T}_k^{i+1}, \mathrm{T}_k^{i-1}$ denote the neighboring rooms. The ambient and heater temperature are $\mathrm{T}_e = -1^\circ$ and $\mathrm{T}_h = 50^\circ$, respectively. The conduction factors are $\alpha_1 =0.45$, $\alpha_2 = 0.045$ and $\alpha_3= 0.09$. The controller $u_i \in [0, 0.5]$,  while the regions of interest are $X_s = [17, 21]^3$ and $X_0 = [18.5, 19.5] \times [18.5, 19.0]^2$. 

For $N = 50$ time-steps, the lower bound on the probabilistic safety for the system is $P_s = 0.95$.
Total CBF and SBF synthesis time is 497.03 seconds for $K=900$ partitions. 
Performing verification on the system, without the controller, yields a safety lower bound of $P_s = 0$.
In Fig.~\ref{fig:room-temp}, for all three rooms, 500 randomly sample trajectories have been propagated for 100 time-steps, under the synthesized control law. As can be observed, none exited the safe set.

\subsection{Nonlinear Unicycle}

Finally, we consider a wheeled mobile robot with the dynamics of a unicycle
\begin{align*}
    \dot{x} = v \cos \theta, \quad
    \dot{y} = v \sin \theta, \quad
    \dot{\theta} = \omega, \quad 
    \dot{v} = a.
\end{align*}
The Cartesian position is defined by \( x \in [-1.0, 0.5] \) and \( y \in [-1.0, 1.0] \). Further, \( \theta \in [-1.75, 0.5] \) represents the orientation relative to the \( x \)-axis, and \( v \in [-0.5, 1.0] \) denotes the speed. The initial set is chosen to be $X_0 \in \mathcal{B}_{\mathbb{R}^4}((-0.4, -0.4, 0, 0), 0.1)$.
The inputs are defined by the steering rate $\omega$ and $a$ as the acceleration. For the discrete-time model, $\Delta t= 0.01$ is used alongside the Euler method. Noise is added to this system, where $\mathbf{w}\sim \N(0, 10^{-4}I)$, capturing the discretization error inherent to the Euler method.

For $N=50$ time-steps, the probability of safety is  $P_s = 0.95$. 
The total optimization time ($K = 2400$ partitions) is 3945.77 seconds. 
The 500 Monte Carlo trajectories for each state are shown in Fig.~\ref{fig:unicycle}, indicating no safety violations.

%% file: IEEE/Sections/6a-2D.tex
\begin{figure*}
    \centering
    \begin{subfigure}[t]{0.30\textwidth}
        \centering
        \includegraphics[width=1\linewidth]{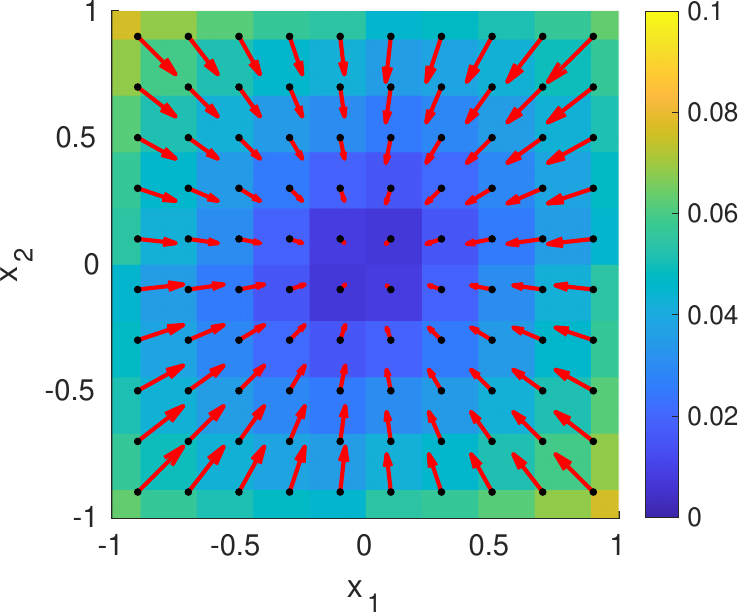}
        \subcaption{Closed-loop vector field. }
        \label{fig:control}
    \end{subfigure}%
    \hspace{5mm}
    \begin{subfigure}[t]{0.33\textwidth}
        \centering
        \includegraphics[width=1\linewidth]{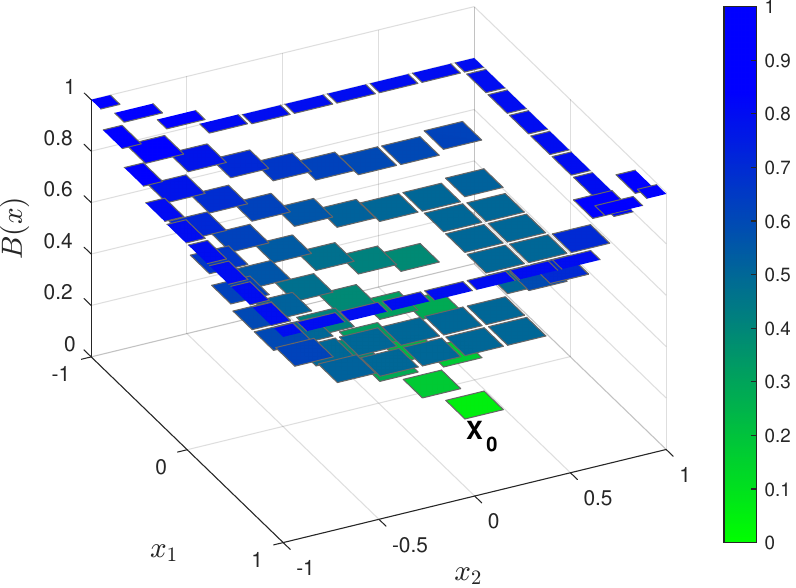}
        \subcaption{PWC s-CBF.}
        \label{fig:barrier}
    \end{subfigure}%
    \hspace{5mm}
    \begin{subfigure}[t]{0.30\textwidth}
        \centering
        \includegraphics[width=\linewidth]{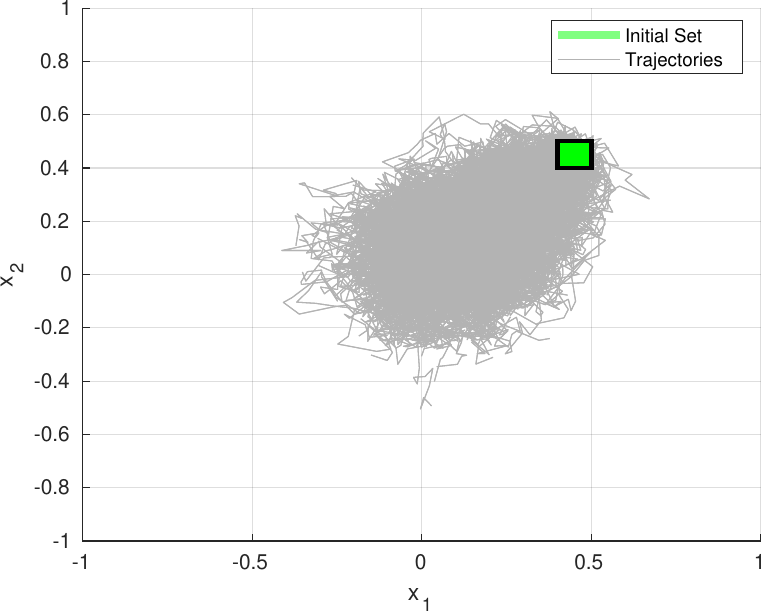} 
        \subcaption{Monte-Carlo Simulations.}
        \label{fig:mc1}
    \end{subfigure}%
    \caption{Convex safe set case study on a linear system. (a) Visualization of the closed-loop vector field, where the background color represents the control magnitude. (b) Representation of the Piecewise Stochastic Control Barrier Function. (c) Results from Monte Carlo simulations, showing 500 trajectories over $N=50$ steps. }
    \label{fig:2d-result}
    \vspace{-3mm}
\end{figure*}

\begin{figure*}
    \centering
    \begin{subfigure}[t]{0.30\textwidth}
        \centering
        \includegraphics[width=1\linewidth]{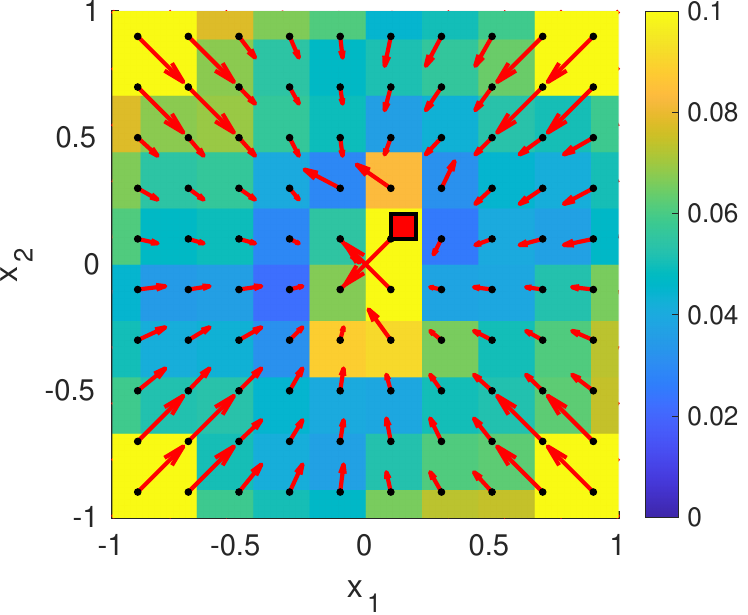}
        \subcaption{Closed-loop vector field.}
        \label{fig:control-nonconvex}
    \end{subfigure}%
    \hspace{5mm}
    \begin{subfigure}[t]{0.33\textwidth}
        \centering
        \includegraphics[width=1\linewidth]{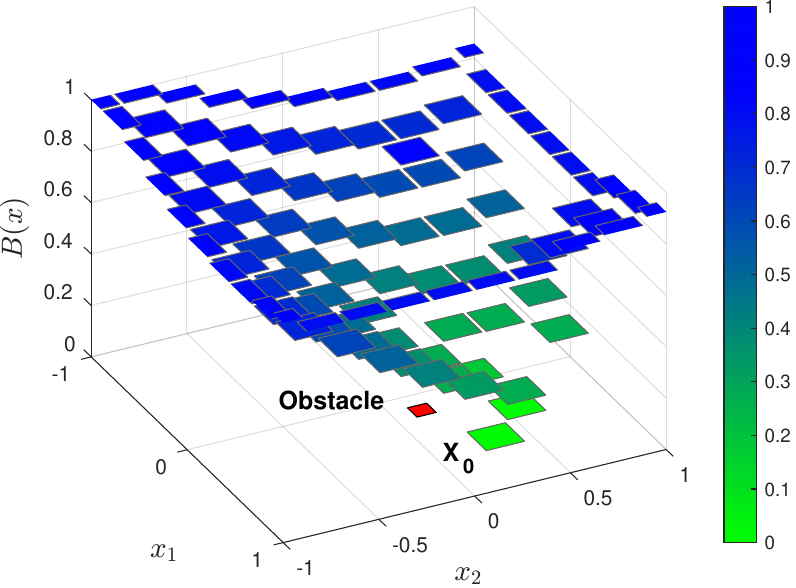}
        \subcaption{PWC s-CBF.}
        \label{fig:barrier-nonconvex}
    \end{subfigure}%
    \hspace{5mm}
    \begin{subfigure}[t]{0.30\textwidth}
        \centering
        \includegraphics[width=\linewidth]{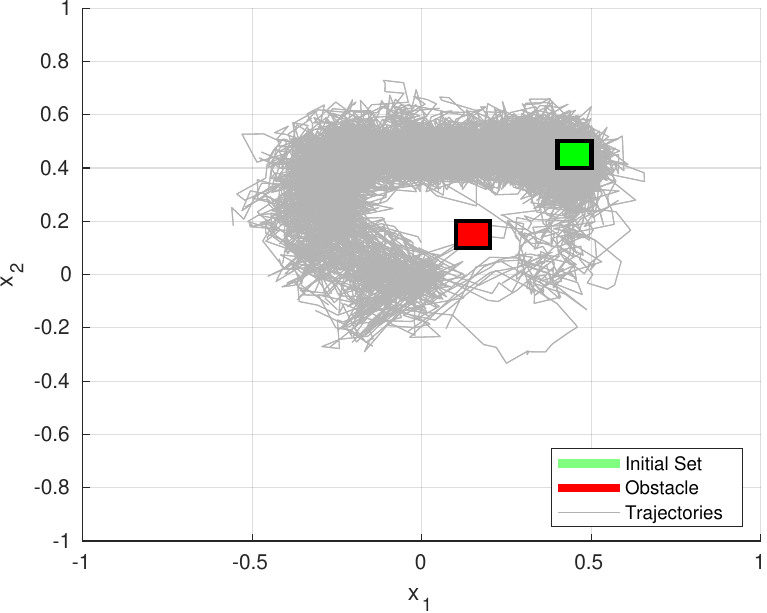} 
        \subcaption{Monte Carlo simulations. }
        \label{fig:mc1-nonconvex}
    \end{subfigure}%
    \caption{Non-convex safe set case study on a linear system. (a) Visualization of the closed-loop vector field, where the background color represents the control magnitude. (b) Representation of the Piecewise Stochastic Control Barrier Function. (c) Results from Monte Carlo simulations, showing 500 trajectories over $N=50$ steps. }
    \label{fig:2d-result-nonconv}
    \vspace{-3mm}
\end{figure*}

%% file: IEEE/Sections/6b-room.tex
\begin{figure*}
    \centering
    \begin{subfigure}[t]{0.31\textwidth}
        \centering
        \includegraphics[width=1\linewidth]{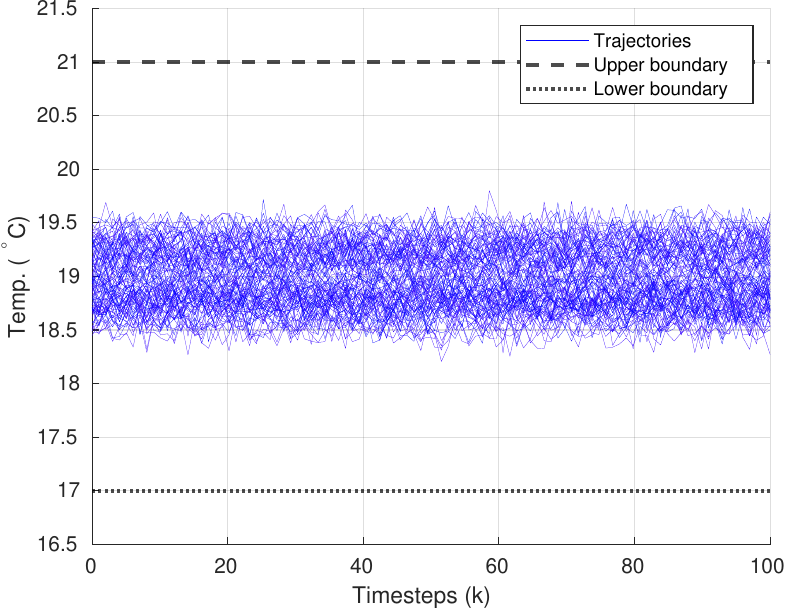}
        \subcaption{Room 1}
        \label{fig:room1}
    \end{subfigure}%
    \hspace{5mm}
    \begin{subfigure}[t]{0.305\textwidth}
        \centering
        \includegraphics[width=1\linewidth]{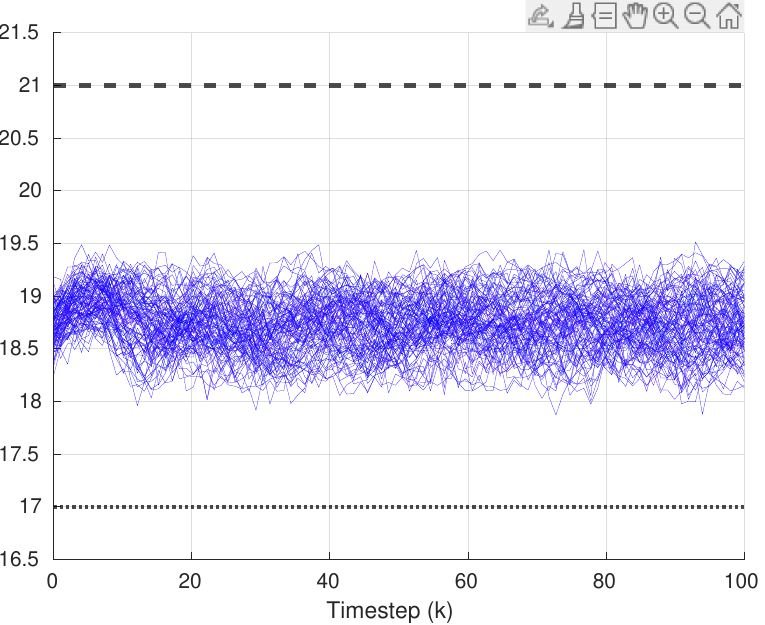}
        \subcaption{Room 2}
        \label{fig:room2}
    \end{subfigure}%
    \hspace{5mm}
    \begin{subfigure}[t]{0.295\textwidth}
        \centering
        \includegraphics[width=\linewidth]{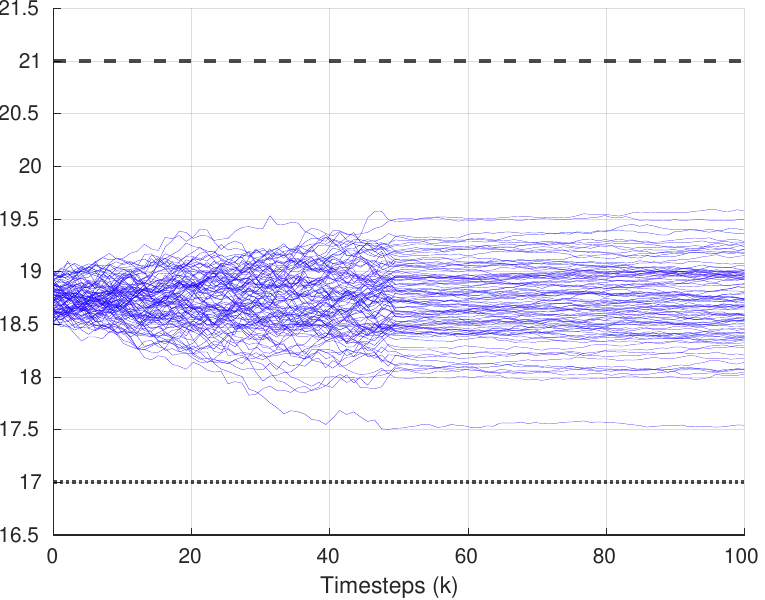} 
        \subcaption{Room 3}
        \label{fig:room3}
    \end{subfigure}%
    \caption{Monte Carlo simulations for the temperature regulation model, showing 500 trajectories over $N=100$ steps. }
    \label{fig:room-temp}
    % \vspace{-3mm}
\end{figure*}

%% file: IEEE/Sections/6c_unicycle.tex
\begin{figure*}
    \centering
    \begin{subfigure}[t]{0.25\textwidth}
        \centering
        \includegraphics[width=1\linewidth]{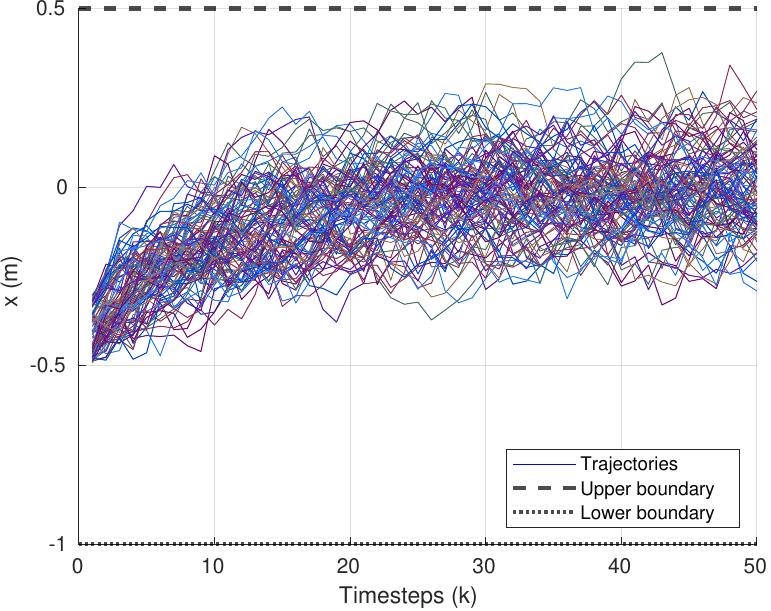}
        \subcaption{Time evolution $x$.}
        \label{fig:unicycle-x}
    \end{subfigure}%
        \begin{subfigure}[t]{0.25\textwidth}
        \centering
    \includegraphics[width=1\linewidth]{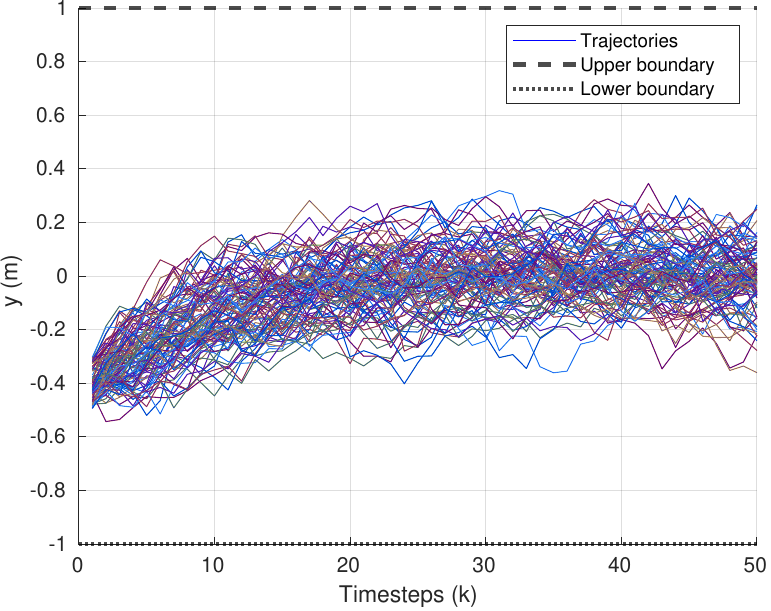}
        \subcaption{Time evolution $y$.}
        \label{fig:unicycle-y}
    \end{subfigure}%
    \begin{subfigure}[t]{0.25\textwidth}
        \centering
        \includegraphics[width=1\linewidth]{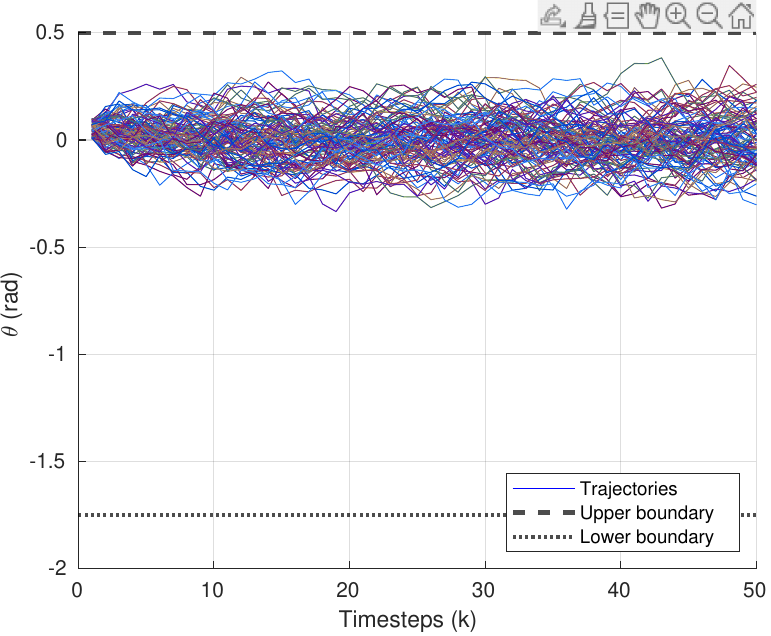}
        \subcaption{Time evolution $\theta$.}
        \label{fig:unicycle-theta}
    \end{subfigure}%
    \begin{subfigure}[t]{0.25\textwidth}
        \centering
        \includegraphics[width=\linewidth]{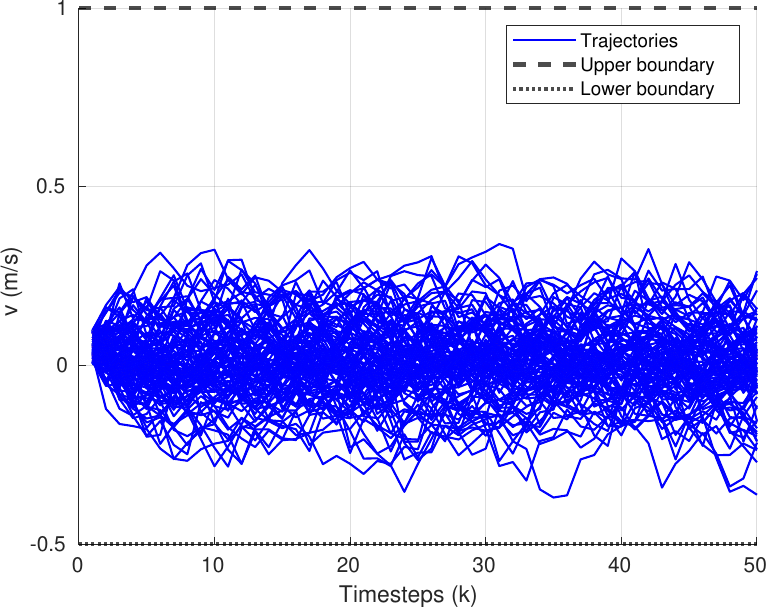} 
        \subcaption{Time evolution $v$.}
        \label{fig:unicycle-v}
    \end{subfigure}%
    \caption{Monte Carlo simulations for the unicycle model, showing 500 trajectories over $N=50$ steps. }
    \label{fig:unicycle}
\end{figure*}

%% file: IEEE/Sections/7_Conclusion.tex
\section{Conclusion}
This paper presents a method for jointly synthesizing a barrier certificate and a safe controller for discrete-time nonlinear stochastic systems. Using piecewise stochastic CBFs, we frame the problem as a minimax optimization, efficiently solved via linear programming. Our approach handles non-additive dynamics with unbounded-support noise. Case studies on linear and nonlinear systems validate its effectiveness in achieving the desired safety probability.

Future work includes improving scalability through faster dual LP solvers, extending to higher-dimensional systems, and incorporating adaptive partitioning for efficiency.
This also includes a more rigorous analysis to guide the selection of the partition size.
Additionally, we aim to minimize control effort in synthesis and integrate reach-avoid objectives for broader applicability in safety-critical settings.